\documentclass[onecolumn, print]{ieeecolor}
\usepackage{generic}
\usepackage{cite}
\usepackage{amsmath,amssymb,amsfonts}
\usepackage{mathtools}
\usepackage{graphicx}
\usepackage{algorithm,algorithmic}
\usepackage{hyperref}
\hypersetup{hidelinks=true}
\usepackage{textcomp}
\let\labelindent\relax
\usepackage{enumitem}
\usepackage{comment}
\usepackage[font=footnotesize]{subcaption} 
\usepackage[font=footnotesize]{caption}

\usepackage{amsthm}
\usepackage{mathtools}
\theoremstyle{definition}
\newtheorem{definition}{Definition}

\theoremstyle{plain}
\newtheorem{assumption}{Assumption}
\newtheorem{lemma}{Lemma}
\newtheorem{theorem}{Theorem}

\newcommand{\gamename}{ToW }
\newcommand{\commalgoname}{Tug-of-Peace }
\newcommand{\nocommalgoname}{Fully Distributed Tug-of-Peace }
\newcommand{\NN}{\mathcal{N}}
\newcommand{\XX}{\mathcal{X}}
\newcommand{\RR}{\mathbb{R}}
\newcommand{\xx}{\boldsymbol{x}}
\newcommand{\yy}{\boldsymbol{y}}
\newcommand{\zz}{\boldsymbol{z}}
\newcommand{\gbold}{\boldsymbol{g}}

\newcommand{\blambda}{\boldsymbol{\lambda}}
\newcommand{\tlambda}{\overline{\lambda}}
\newcommand{\btlambda}{\boldsymbol{\overline{\lambda}}}
\newcommand{\bdelta}{\boldsymbol{\delta}}
\newcommand{\intXX}{\XX^{\textrm{o}}}
\newcommand{\expec}{\mathbb{E}}
\newcommand{\FF}{\mathcal{F}}
\newcommand{\GG}{\mathcal{G}}
\newcommand{\MM}{\mathbf{M}}
\newcommand{\II}{\mathbb{I}}
\newcommand{\AAA}{\mathcal{A}}
\newcommand{\ZZ}{\mathcal{Z}}

\def\BibTeX{{\rm B\kern-.05em{\sc i\kern-.025em b}\kern-.08em
    T\kern-.1667em\lower.7ex\hbox{E}\kern-.125emX}}
\begin{document}
\title{Choose Your Battles: Distributed Learning Over Multiple Tug of War Games}
\author{Siddharth Chandak, Ilai Bistritz, Nicholas Bambos
\thanks{\hrule\medskip Siddharth Chandak and Nicholas Bambos are with the Department of Electrical Engineering, Stanford University, CA, USA. Ilai Bistritz is with the School of Industrial and Intelligent Systems Engineering and the School of Electrical and Computer Engineering, Tel Aviv University, Israel. This research was supported by the Koret Foundation grant for Smart
Cities and Digital Living. Emails: chandaks@stanford.edu, ilaibistritz@tauex.tau.ac.il, bambos@stanford.edu}
}

\maketitle

\begin{abstract}
Consider $N$ players and $K$ games taking place simultaneously. Each of these games is modeled as a Tug-of-War (ToW) game where increasing the action of one player decreases the reward for all other players. Each player participates in only one game at any given time. At each time step, a player decides the game in which they wish to participate in and the action they take in that game. Their reward depends on the actions of all players that are in the same game. This system of $K$ games is termed a `Meta Tug-of-War' (Meta-ToW) game. These games can model scenarios such as power control, distributed task allocation, and activation in sensor networks. We propose the Meta Tug-of-Peace algorithm, a distributed algorithm where the action updates are done using a simple stochastic approximation algorithm, and the decision to switch games is made using an infrequent 1-bit communication between the players. We prove that in Meta-ToW games, our algorithm converges to an equilibrium that satisfies a target Quality of Service reward vector for the players. We then demonstrate the efficacy of our algorithm through simulations for the scenarios mentioned above. 
\end{abstract}

% \begin{IEEEkeywords}
% Enter key words or phrases in alphabetical 
% order, separated by commas. For a list of suggested keywords, send a blank 
% e-mail to keywords@ieee.org or visit \underline
% {http://www.ieee.org/organizations/pubs/ani\_prod/keywrd98.txt}
% \end{IEEEkeywords}

\section{INTRODUCTION}
In many network scenarios, a conflict arises between the agents regarding how to share resources and tasks. Examples are transmitter power control in wireless networks, activation probability in sensor networks, and task allocation in multi-robot networks. Each agent has a local performance measure, such as its throughput, energy consumption, or progress rate. When a wireless device uses more power to transmit, it increases its throughput but decreases that of others. When a sensor, which collects its own data but also relays that of others, is activated with a lower probability, it improves its energy consumption but worsens the packet loss rate of others. When a robot spends effort on a task, it increases its progress rate but decreases the marginal contribution of others, possibly pushing them to switch to a task where they can be more efficient. These conflicts can naturally be modeled as a game between the agents. 

A conflict between agents, however, does not mean that agents are selfish. The agents are programmable and run the protocol we design (e.g., WiFi or 5G). Moreover, the agents (or players) often have a minimum requirement of reward, which we term as their Quality of Service (QoS). 

Finding an action profile that satisfies the QoS guarantees is an optimization problem that can be solved by a centralized server that instructs the agents how to play. However, such a centralized server would have to know the reward functions and action sets, which implies that agents must send this information to the server. The server then needs to compute the solution in real-time and send it back to the agents. This creates latencies and limits the system's scalability, since the complexity of the optimization problem grows fast with the number of agents. Such a centralized scheme also violates user privacy and creates a security vulnerability, since if the server is hacked, the hacker gains control of all agents. Therefore, distributed protocols are preferred \cite{power_foschini,power_zhang,power_eigen_3, power_ulukus,power_yates, mandal2025distributed, qu2019distributed}.

We wish to design a distributed algorithm that allows each player to achieve their QoS asymptotically. The challenge in designing such a protocol is that players often just observe \textit{stochastic bandit feedback}, i.e., a noisy version of the reward at the current action profile. The players do not know their reward functions and do not observe the actions or rewards of others. Hence, the challenge is to know whether to ``insist'' on improving their local reward or give up in favor of others. In large-scale networks with thousands or millions of devices, this coordination challenge is the primary design bottleneck. 

In this paper, we present a \textit{meta game} where each player decides the game they wish to participate in \textit{and} the action they take in that game. In any particular game, increasing the action of one player decreases the rewards of other players in the same game. We name this class of games \textit{Tug-of-War (ToW)} games, since they resemble players pulling a rope in opposite directions. We show that transmission power control, tasks with marginally diminishing utilities, and activation in sensor networks can be modeled as ToW games (Section \ref{sec:app}). We term the multiple game system a \textit{Meta-ToW} game, owing to the `meta' level where players decide which game to join. This can model scenarios with orthogonal resources or tasks. These Tug-of-War games capture a common interaction in networking applications and allow distributed and scalable coordination for players seeking to learn QoS guarantees. 

We first design the Tug-of-Peace algorithm (Algorithm \ref{commalgo}), a simple distributed algorithm for a ToW game. This algorithm converges with probability 1 to an action profile that satisfies the given QoS requirements for all players, if the QoS requirements are feasible. Additionally, we show that with high probability, Algorithm \ref{commalgo} converges to the \textit{minimal action profile} (Definition \ref{def:minimal}) which satisfies the QoS requirements. This algorithm requires infrequent (i.e., finitely many), one-bit communications. We also present an algorithm (Algorithm\ \ref{nocommalgo}) for the case where even one-bit communication is not possible. We show that this algorithm also converges to the minimal action profile with high probability. We name these algorithms \textit{Tug-of-Peace (ToP)} since, instead of pulling the rope as hard as they can, players learn to place the marker at a point where they can all achieve their target QoS. 

 In the meta game, the players need to find a configuration of players to different games such that the QoS requirements are feasible in that configuration. We design the Meta-ToP algorithm (Algorithm \ref{meta-algo}), a distributed algorithm in which players use the infrequent one-bit communication to communicate to other players that the QoS requirement `may' not be feasible in the current configuration of players. On receiving these communications, the players switch to different games, allowing players to try different configurations. We show that our algorithm attains a configuration and converges to an action profile that satisfies the QoS requirements for all players with probability 1. Moreover, we show that the game switches happen only a finite number of times. Finally, we demonstrate in simulations the performance of our algorithms for power control, distributed task allocation and sensor activation.  

\subsection{Related Work}

Distributed algorithms that converge to a Nash equilibrium (NE) have been extensively studied for various games \cite{ye2017distributed,gadjov2018passivity,lin2023statistical,frihauf2011nash,tatarenko2020geometric}. The NE can be globally inefficient, and in particular, does not satisfy any QoS guarantees. Instead of a NE, the equilibrium our players seek in this paper is an action profile where all of their rewards are above the QoS threshold.

Distributed protocols with QoS guarantees for general games are much less studied compared to distributed algorithms that converge to NE. The work in \cite{bistritz2021one} considered a multiplayer bandit scenario where, if multiple players pick the same arm, they all receive zero reward. This is a special case of a game that is different from the class of games we consider.  The work in \cite{bistritz-queue} considers a general discrete game, but the algorithm uses regular communication between the players. Compared to both algorithms in \cite{bistritz2021one, bistritz-queue}, our proposed algorithm here has substantially better scalability in the number of players. Furthermore, our algorithm has two variants that either require 1-bit of infrequent communication or no communication at all. QoS guarantees were also studied in \cite{Borkar-QoS}, but in the context of average cost setting in a multi-agent Markov decision process.

Distributed protocols with QoS guarantees are more common for power control in wireless networks, specifically \cite{power_foschini,power_zhang,power_eigen_3, power_ulukus,power_yates}. The power control game is a special case of the Meta-ToW games we consider in this paper, as explained in Subsection \ref{app-power}. The power control QoS algorithms exploit the form of the reward function, which they assume is known. Our algorithm, on the other hand, does not need to know the reward functions as long as the game is a ToW game. 

The results of \cite{power_ulukus} provide an algorithm for power control, where assumptions are made on the knowledge available to the players. The algorithm needs each player to know the channel gain between itself and its intended receiver. \cite{power_foschini} removes this assumption but fails to converge in the presence of noise (as shown by \cite{power_zhang}). The results in \cite{power_zhang, power_eigen_3} are the closest to our algorithm as they are also stochastic approximation algorithms that can handle noise and do not assume further knowledge of the system. But they have different assumptions as they require an unbiased estimator of the inverse reward, i.e., the inverse of the Signal-to-Interference and Noise Ratio (SINR).

The analysis methods of \cite{power_foschini,power_zhang,power_eigen_3, power_ulukus} are specific to the power control wireless scenario. Their analysis depends on the eigenvalues of a matrix defined using the channel gain matrix and the QoS vector. As an example, \cite{power_zhang} models the problem as finding the solutions of a linear system, and its analysis then depends on the eigenvalue properties of this linear system. This model and properties are specific to that scenario and cannot be extended to our broader class of games. Other works, such as \cite{power_yates}, take a general approach, but also with certain limitations. They use the properties of monotone and sub-homogeneous maps to show the convergence of the algorithm they design. These properties are facilitated by the structure of power control games. However, designing an algorithm that satisfies these properties is challenging for other general games. Additionally, our algorithm is better equipped to deal with noise than the algorithm by \cite{power_yates}, as we show later. 

% The first series of works \cite{power_zhang,power_foschini} give an algorithm similar to ours, but their analysis method is specific to the power control wireless scenario. They model their problem as finding the solution of a linear system and their analysis then depends on the eigenvalue properties of this linear system. This model and properties are specific to their scenario and cannot be extended to our broader class of games. 

% Other works such as \cite{power_yates} take a general approach, but not without limitations. They use the properties of monotone and sub-homogeneous maps that are specific to power control and are not generalizable to our work. Our algorithm is also better equipped to deal with noise than that in \cite{power_yates}. Another distinction is that all these algorithms for QoS in power control exploit the form of the reward function, which they assume is known. A known form for a reward function is a reasonable assumption when designing an algorithm for a specific application. Our algorithm, on the other hand, does not need to know the reward functions, as long as the game is a "Tug-of-War" game. 

A key step in the analysis in \cite{power_zhang,power_foschini,power_yates} is that their assumptions guarantee the existence of a unique equilibrium that satisfies $u_n(\xx)=\lambda_n$ for all $n$, where $u_n(\xx)$ is the reward received by player $n$ at action profile $\xx$ and $\lambda_n$ is their QoS requirement. Furthermore, this point is guaranteed to be globally asymptotically stable. However, the existence of a unique and globally asymptotically stable equilibrium is not guaranteed in general ToW games. Moreover, the action sets in a ToW game are bounded, which gives rise to undesirable equilibria at the boundaries. Due to the observation noise, it is tricky to guarantee which equilibrium the algorithm would converge to. As we later see in Section \ref{sec: algo}, we give guarantees on convergence, and results which show that our algorithms converge to the ``best'' equilibrium with high probability. 

Given the general analysis and because we deal with unknown reward functions, our algorithm and its guarantees carry over to applications beyond power control with no modifications needed. One such example is activation in sensor networks, which we study in Subsection \ref{app-acti}. In this application, the observations are noisy and it is unreasonable to assume knowledge of the reward functions.

A Meta-ToW game consists of multiple games that players can join and leave. Thus, our work is related to the growing literature on open multi-agent networks\cite{deplano2025optimization,bistritz2024gamekeeper,liu2025online}.

Our analysis employs the Ordinary Differential Equation (ODE) approach to stochastic approximation \cite{Borkar-book} and the concept of Cooperative ODEs and Monotone Dynamical Systems \cite{coop-review}. Other than in dynamical systems, cooperative ODEs have been widely used in fields such as epidemiology \cite{coop_use_epi_1,coop_use_epi_2}, social networks \cite{coop_use_social_1,coop_use_social_2}, and queuing systems \cite{Borkar-queue}.

Our preliminary conference paper  \cite{chandak2023tug} studied the special case of a single ToW game. This extended paper includes the analysis that was all omitted from \cite{chandak2023tug}. Moreover, we generalize the scope of \cite{chandak2023tug} to Meta-ToW games where each player chooses both the game they wish to participate in and the action they take. This allows for modeling multichannel power control and task allocation, going beyond \cite{chandak2023tug}.

\subsection{Notation}
We use bold letters to denote vectors and matrices and $\mathbf{0}_N$ and $\bdelta_N$ to denote the $N$-dimensional all-zero and all-$\delta$ vectors, respectively. We use $\Pi_\XX$ to denote the Euclidean projection into the set $\XX$. We use the standard vector inequalities, where $\xx\leq\yy$ denotes $x_i\leq y_i$ for all $i$ and  $\xx\in[\zz,\yy]$ denotes $z_i\leq x_i\leq y_i$ for all $i$. We use the standard game-theoretic notation where $\xx_{-n}$ is the vector of actions for all players except player $n$. $\|\cdot\|$ denotes the Euclidean norm when applied on a vector, and the operator norm for matrices (under the Euclidean norm for vectors). 

\section{Problem Formulation}
Consider a set of $N$ players $\NN=\{1,\ldots, N\}$ and a set of $K$ games $\GG=\{1,\ldots,K\}$. Each player $n$ decides the game $g_n\in\GG$ they wish to participate in and the action $x_n\in\XX_n$ they take in that game. Here $\XX_n\coloneqq [0,B_n]\subseteq{\RR^+}$. Let $\XX=\XX_1\times\ldots\times\XX_N$. Define the interior of this region as $\intXX= (0,B_1)\times\ldots\times(0,B_N)$. Each player $n$ has a reward function $u_n(\gbold,\xx)$ that depends only on the actions of the players in the game $g_n$. Here $\gbold\in\GG^N$ is the game configuration vector, which indicates the game each player is in, i.e., $\gbold=[g_1,\ldots,g_N]$ and $\xx$ is the action profile for all players, i.e., $\xx=[x_1,\ldots,x_N]$.

Define the Quality of Service (QoS) vector $\blambda=[\lambda_1,\ldots,\lambda_N]$. Our goal is to design a distributed algorithm such that the players, under \textit{noisy bandit feedback}, converge to a game configuration and action profile such that
\begin{equation}
    u_n(\gbold, \xx) \geq \lambda_n,\quad \forall n \in \mathcal{N}. \label{eq:1}
\end{equation}
We next explain the noisy bandit feedback and the required assumptions on the QoS vector.

Each player chooses one game $g_n(t)$ to be in, and their action $x_n(t)$ in this game, at each timestep $t\geq 0$. The player then receives a noisy reward $y_n(t)=u_n(\gbold(t),\xx(t))+M_n(t)$. For example, in communication systems, the reward could be the probability to successfully transmit a packet, while the noisy feedback is ACK/NACK. The players receive no information about other players' games, actions, and rewards. They also cannot observe the rewards they would have received in other games. Here $\MM(t)=[M_1(t),\ldots,M_N(t)]$ denotes the noise sequence and satisfies the following assumption.
\begin{assumption}\label{assu-martingale}
    $\MM(t)$ is a martingale difference sequence, i.e., $\expec[\MM(t)|\FF_{t-1}]=0$ where $\FF_{t}\coloneqq \sigma(\xx(s), \gbold(s), \MM(s), s\leq t)$. Additionally, $\MM(t)$ has bounded support, i.e., $|M_n(t)|\leq \widehat{M}$, w.p.\ $1$, for all $n\in\NN, t\geq0$ for some positive $\widehat{M}$.
\end{assumption}

Not all QoS vectors $\blambda$ are feasible, and hence we make the following feasibility assumption.
\begin{assumption}\label{assu-delta}
We make the following assumptions.
    \begin{enumerate}[label=(\alph*)]
        \item There exist $\xx\in\intXX, \gbold\in\GG^N$ and $\delta>0$ such that $u_n(\gbold,\xx)\geq \lambda_n+\delta$ for all players $n$.
        \item Let $\xx_g$ denote the vector of actions of players that are in game $g$ and $Du(\xx_g)$ denote the Jacobian at $\xx_g$, i.e., the Jacobian for players' actions that are in game $g$. Suppose $\btlambda$ is chosen uniformly at random in $[\blambda,\blambda+\bdelta_N]$. If $u(\hat{\gbold},\hat{\xx})=\btlambda$, then w.p.\ 1 the Jacobian $Du(\hat{\xx}_g)$ has no purely imaginary eigenvalues for all $g\in\GG$, where the game configuration is given by vector $\hat{\gbold}$.
    \end{enumerate}
\end{assumption}

Part (a) of this assumption ensures that the QoS vector is feasible and can be achieved in the interior of the set $\XX$. Here $\delta$ can be infinitesimally small. 

The parameter $\delta$ is also needed for technical reasons. For an arbitrary QoS vector $\blambda$, the QoS might be achieved at an equilibrium point that is not stable for our algorithm \cite{hirsch_smale}. We later show in Lemma \ref{lemma:stable} that under Assumption \ref{assu-delta}, the corresponding equilibrium point is stable with probability $1$ when the QoS vector is drawn uniformly at random from $[\blambda,\blambda+\bdelta_N]$. 

A key requirement for this result is part (b), which is a mild nondegeneracy assumption. Intuitively, this condition is not restrictive: for non-zero purely imaginary eigenvalues, an arbitrarily small perturbation typically introduces a non-zero real part \cite{Borkar-queue}, \cite[Section 8.4]{hirsch_smale}. Thus, for a randomly chosen $\btlambda$, one would generally expect the Jacobian $Du(\hat{\xx}_g)$ not to have purely imaginary eigenvalues.

\subsection{Meta Tug-of-War Games}
In general, a distributed algorithm that solves the QoS problem defined above may require significant coordination among players \cite{bistritz-queue}. Remarkably, we identified a wide class of games where a simple distributed algorithm provably converges to an action profile with QoS guarantees. We term this class of games Meta-ToW, which naturally models common resource conflicts in a network of agents. We start by defining a single ToW game. As there is only a single game, we use $u_n(\xx)$ to denote the reward for player $n$.
\begin{definition}\label{def-game}
A game is a \textbf{Tug-of-War (ToW) game} if for all players $n\in\NN$, the reward function $u_n(\cdot)$ is continuously differentiable and satisfies the following condition:
\begin{equation}
\frac{\partial u_n(\xx)}{\partial x_m}<0, \forall \; m\neq n\in{\NN}, \xx \in\intXX.
\end{equation}
Furthermore, we assume that for all $n$, $u_n(\xx)=0$ if $x_n=0$ and $u_n(\xx)\geq0, \forall n\in\NN, \xx\in\XX$.
\end{definition}
Intuitively, these are games where if a player increases their action keeping all else constant, then the reward of other players would drop. A broad class of games satisfying this condition are resource allocation games \cite{Agrawal, Bistritz_res_allo}. In such games, the action taken by a player is the amount of resource used by that player; if a player increases their resource use then other players' reward would drop if they maintain their resource use. Definition \ref{def-game} makes no assumption on the impact of a player's reward upon changing their own action (i.e., no assumption on $\partial u_n/\partial x_n$). This flexible definition can easily model many practical games, as we detail in Section \ref{app-acti}. 

There could be multiple action profiles that meet the QoS guarantees. Since in ToW games the player's action is a measure of ``pull'' (e.g., power), we prefer equilibrium points with smaller actions.

We now define Meta-ToW games. 
\begin{definition}\label{defn:meta-ToW}
    In a \textbf{Meta Tug-of-War (Meta-ToW)} game, for all players $n\in\NN$, the reward function is continuously differentiable in $\xx$ and satisfies the following.
    \begin{itemize}
        \item For $m\neq n$, if $g_m=g_n$, i.e., they are in the same game,
            $$\frac{\partial u_n(\gbold,\xx)}{\partial x_m}<0, \xx\in\intXX$$
        \item For $m\neq n$, if $g_m\neq g_n$, i.e., they are in different games, 
            $$\frac{\partial u_n(\gbold,\xx)}{\partial x_m}=0, \xx\in\XX.$$
    \end{itemize}
    The reward function also satisfies $u_n(\gbold,\xx)=0$, if $x_n=0$ and $u_n(\gbold,\xx)\geq 0$ for all $n,\xx\in\XX$ and $\gbold\in\GG^N$. 
\end{definition}
 
\section{Applications}\label{sec:app}
In this section, we detail three useful special cases of Meta-ToW games and the QoS achievability problem.

\subsection{Multi-Channel Power Control in Wireless Networks}\label{app-power}
In power control, the players are $N$ transmitter-receiver pairs, so transmitter $n$ wishes to transmit to receiver $n$. There are $K$ channels, and at each timestep, each player $n$ decides the channel $g_n$ to transmit over. For each player $n$, the action $x_n$ denotes the transmission power of player $n$ in channel $g_n$. The interference experienced by receiver $n$ is given by $$I_n(\gbold, \xx)=\sum_{m\neq n: g_m=g_n} c_{m,n}x_m,$$ where $c_{m,n}>0$ is the channel gain between the transmitter of player $m$ and receiver of player $n$. Each receiver also faces additive Gaussian noise with variance $N_0$. Then the utility of player $n$ is the Signal-to-Interference Ratio (SINR) given by:
\begin{equation}
    u_n(\gbold,\xx)=\frac{c_{n,n}x_n}{N_0+I_n(\gbold, \xx)}.
\end{equation}

A common objective of power control in wireless networks is to find the minimum transmission power (or action profile) for each player such that all players satisfy their QoS requirement. Centralized control is infeasible due to issues such as latency, communication overhead, and lack of infrastructure. 

In a distributed setting, each player can only obtain a noisy estimate of their SINR \cite{SIR-est}. The estimation noise can be modeled as Martingale difference noise, as shown in \cite{power_zhang}. This game is a ToW game since the utility function satisfies:
$$\frac{\partial u_n(\gbold, \xx)}{\partial x_m}=-\frac{c_{n,n}c_{m,n}x_n}{(N_0+I_n(\gbold, \xx))^2}<0$$
for all $m\neq n$ such that $g_m=g_n$ and $\xx\in\XX^o$. There are multiple different assumptions on $\blambda$ and the matrix $C=[c_{m,n}]$ in the literature \cite{power_zhang, power_foschini} which all result in the existence of a point that satisfies the QoS requirement. Taking the boundary $B_n$ to be large enough for each player $n$ would ensure that the QoS is achieved in the bounded region $[0,B_n]$. Moreover, in the case of a single channel (or equivalently a single game in our framework) these assumptions also result in the uniqueness of the equilibrium point $\hat{\xx}$ which satisfies $u_n(\xx)=\lambda_n$ for all players $n$. Trivially, this point is also the minimal equilibrium point.

In wireless networks, the locations of the devices, and thus the channel gains, are typically modeled as random \cite{power_foschini}. Then, there is no need for Assumption \ref{assu-delta} as the randomness in the channel gain matrix $C$ has a similar effect on the Jacobian's eigenvalues.  
% We also wish to point out that \cite{power_foschini} assumes some randomness in the matrix $G$, which stems from the players' random placement. Under certain conditions, such randomness in $G$ leads to some properties of eigenvalues of $G$ which guarantee the almost sure stability of the equilibrium point. Hence under this assumption, the players do not need to sample $\tlambda_n$ at random, since the scenario itself is random.

\subsection{Distributed Task Allocation with Cooperative Agents}
Consider $N$ agents who collaborate to perform $K$ tasks \cite{mandal2025distributed,qu2019distributed}. Each player decides which task to work on and how much effort they make on this task. The total utility obtained from a task depends on all the agents' efforts working on that task, which can be thought of as the value of the progress rate of this task. This total utility is then divided between the agents working on that task, which represents their marginal contribution to the total value attributed to the progress in the task. This total utility is monotonically increasing and concave with respect to the total effort spent on that task. Progress rate typically increases sublinearly in effort or linearly at best if there is no overhead in collaboration. Then, the value of increased progress usually has diminishing returns. For example, computing a job in one week instead of two is a much more significant boost than computing a job in one minute instead of two. This reward structure incentivizes agents to choose tasks where their marginal contribution is more significant. This reward, which measures the local progress of an agent, can be naturally measured locally, as opposed to the global progress. Examples for applications are LLM agents, cloud computing, and multirobot manufacturing or construction.  

As a concrete example, let the total utility obtained for task $g\in\{1,\ldots,K\}$ be $$U_g(\gbold,\xx)=\log\left(\alpha_g+\sum_{m\in\NN, g_m=g}\beta_{m,g}x_m\right),$$
for some constant $\alpha_g\geq 1$, and $\beta_{m,g}\geq 0$. The constant $\beta_{m,g}$ denotes the proficiency of agent $m$ in performing task $g$. The reward for player $n$ is then given by
\begin{equation}\label{utility-tasks}
    u_n(\gbold,\xx)=\left(\frac{\beta_{n,g_n}x_n}{\sum_{m\in\NN, g_m=g_n}\beta_{m,g_m}x_m}\right)U_{g_n}(\gbold,\xx),
\end{equation}
which is the agent's relative part in the overall task progress rate. When the agent increases its action, the rewards of others decrease (they now contribute less to the new progress rate). Thus, the above game is a Meta-ToW game. The QoS requirements allow the operator to ensure that all robots are utilized efficiently, as their marginal contribution is guaranteed. The sum of the QoS vector guarantees the global progress rate. 

\subsection{Activation in Sensor Networks}\label{app-acti}

Consider $N$ sensors that communicate over a wireless network. Each sensor collects data and transmits these observations to a destination through the communication network, i.e., they only communicate with their neighbors on the network. Hence, the sensors have a dual role - observing their surroundings and relaying the observations they receive from other sensors. Players wish to collect as much data as possible, but also to reduce energy consumption to save their batteries. So the sensor is active at any time only with a given probability. When activated, sensors both make observations and relay other observations. Each sensor wishes to find a probability of activation that balances between collecting more data and reducing energy consumption. Problems with similar formulations are studied in \cite{sensor_1,sensor_2}.

Let the $N$ sensors be the players, where the probability of player $n$ being active is $p_n$. The action taken by player $n$ is $x_n=1-p_n$, i.e., the probability of being off. Clearly, $x_n$ is in $[0,1]$. We define the reward function for each player $n$ as 
$$u_n(\xx)=f(P_n(\xx))-\alpha+\beta x_n.$$
Here $P_n(\xx)$ is the probability that player $n$'s observation was successfully transmitted to the destination. This probability depends on the actions of all players as player $n$'s action affects the number of observations it makes, and others' actions affect how many packets are relayed to the destination. $f(\cdot)$ is the monotonically increasing `value' assigned to observations, which is typically concave since the marginal value of data decreases as the amount of data received increases. Finally,  $\beta x_n$ is the `reward' for consuming less energy thanks to being off with probability $x_n$. This can be thought of as $\beta-\beta p_n$, which is a shifted cost due to battery usage. Then $-\alpha<0$ is just an offset term that ensures that $u_n(\xx)=0$ if $x_n=0$.

$P_n(\cdot)$ depends on the communication network and can be complex. However, it is easy to justify that this game is a \gamename game with $K=1$. When player $m$ increases their action, it is off with a higher probability, which decreases the times it relays the observations of player $n$ and hence decreases $P_n(\cdot)$. 

Let $L$ be the number of data packets an active sensor transmits between two turns of our game (i.e., decision periods). In general, $L$ can be stochastic and different between sensors, but we assume it is constant for simplicity. At the end of each timestep, sensors get feedback about how many of the $L$ observation packets were successfully received, which provides them with a noisy unbiased estimator of $P_n(\xx(t))$.

 %We assume that players make decisions at timesteps $t=0,1,\ldots$. The interval $[t,t+1]$ is %divided into $L$ further steps, and sensor $n$ chooses to remain off at each of these steps with probability $x_n(t)$. 

Unlike in previous applications, here the action $x_n$ cannot be thought of as `effort' or the amount of resource used. The reward $u_n$ is not necessarily monotone in $x_n$ over $[0,1]$. Clearly $P_n(\xx)$ decreases and $\beta x_n$ increases with increasing $x_n$. Hence we can expect a peak at some point if we keep $\xx_{-n}$ fixed and change $x_n$. The aim of the players in this game is to collect as much data as possible but they are limited by their energy consumption. The QoS requirement ensures that sufficient observations are obtained from all sensors. The ``best'' equilibrium in this example is one where they collect the most data, which is given by the minimal equilibrium.

\section{Tug-of-Peace Algorithms}\label{sec: algo}
In this section, we present algorithms to achieve QoS guarantees for our set of games. 

We first present two algorithms for a single Tug-of-War game. The Tug-of-Peace (ToP) algorithm converges with probability 1 to a point where each player $n$ receives at least reward $\lambda_n$, but it requires $1$-bit communication between players at some timesteps. We later show that this communication is required only a finite number of times. The Fully Distributed Tug-of-Peace (FDToP) algorithm is for the case where strictly no communication is possible. The guarantees suffer in this case, but we give a result on convergence to an action vector that satisfies the QoS requirement with high probability. 

We then present the Meta-ToP algorithm for the setting of Meta-ToW games. We borrow intuition from the single game setting and use the $1$-bit communication for players to switch between game configurations till they find a configuration under which the QoS requirements are feasible. We show that in a finite number of game switches, our algorithm converges to a game configuration and action profile that satisfy the QoS requirement with probability 1. 

\subsection{\commalgoname Algorithm}
We now give an intuitive explanation for the workings of Algorithm \ref{commalgo}. 

Each player first samples a $\tlambda_n$ uniformly at random from the narrow interval $[\lambda_n,\lambda_n+\delta]$. This randomization is just a technical step and is done to ensure the almost sure stability of equilibrium points. Note that this does not affect the QoS requirement, as any point that satisfies $u_n(\xx)=\tlambda_n\geq\lambda_n$ satisfies the QoS requirement for that player. 

Now, let us discuss the iteration performed by each player. For simplicity, first consider a case where there is no noise, so each player observes their exact reward, i.e., $y_n(t)=u_n(\xx(t))$. The idea behind the algorithm is that each player $n$ starts with action $0$ and, upon receiving a reward below their QoS requirement, increases their action. This increase is at a rate proportional to its `dissatisfaction', i.e., how far its reward is from the QoS requirement. By definition of \gamename games, this causes a drop in the rewards received by other players, which encourages them to increase their actions. This `cooperative' increase in actions eventually leads to convergence to an equilibrium point. 

Even if there exists a feasible point in the interior, the noise can cause a player to reach the boundary and get stuck there. To avoid this, when a player tries to exceed the boundary, it sends a signal to other players stating that it might be stuck at the boundary. On receiving such a signal, all players return to the action $0$. This \textit{resets} the iteration and brings it to a point that is in the domain of attraction of an equilibrium where the QoS condition is satisfied. With a lower starting stepsize after the reset, as required by Assumption \ref{assu-step}, the probability of reaching the boundary due to noise is lower.

The signal that indicates a player is stuck at the boundary carries $1$-bit of information. We show that these resets happen only a finite number of times with probability 1. Hence, the communication overhead of Algorithm 1 is negligible. Implementing this communication is application-dependent. For example, communication between devices is more obvious in sensor networks than in power control and task allocation. If agents can only talk with their neighbors on a graph, then there will be a delay of some turns until all players receive these signals, but nothing else changes in our analysis or results.

% If there exists no such equilibrium point in the interior of $\XX$, i.e., the vector $\blambda$ is not achievable and Assumption \ref{assu-delta} is not satisfied, then all players keep increasing their actions and eventually reach a point where one or more players play their action at boundary. 

We need to choose appropriate stepsizes to deal with the noise and the resets, specified in the following standard assumption:
\begin{assumption}\label{assu-step}
    The stepsize sequence $0<\eta(t)<1$ for $t\geq 0$ satisfies the following:
    $$\sum_t \eta(t)=\infty,\sum_t \eta(t)^2<\infty\;\text{and}\;\eta(t+1)<\eta(t)\;\forall t.$$
\end{assumption}

There may exist {\bf multiple equilibria} which satisfy the QoS requirement. They are not all equally desirable since some require more ``pull'' from the players than others (e.g., more energy or power). Our algorithm selects with high probability the ``minimal equilibrium'', in the following sense:  
\begin{definition}\label{def:minimal}
    $\xx_*$ is the \textbf{minimal equilibrium} if among all equilibrium points $\hat{\xx}$ which satisfy $u_n(\hat{\xx})=\tlambda_n$ for all $n$, $\xx_*$ is the smallest component-wise, i.e., $\xx_{*_n}\leq \hat{\xx}_n$ for all $n$.
\end{definition}
The existence of such an equilibrium point is guaranteed by Lemma \ref{lemma:coop} in the next section. 
The next result gives convergence guarantees for the \commalgoname algorithm:
\begin{theorem}\label{main-thm-comm}
Under assumptions \ref{assu-martingale}-\ref{assu-step}, the following statements hold:
\begin{enumerate}[label=(\alph*)] 
    \item With probability 1, the iterates of Algorithm \ref{commalgo} converge to an equilibrium point $\hat{\xx}$ which satisfies $u_n(\hat{\xx})=\tlambda_n\geq\lambda_n$ for all $n$. Moreover, the reset to $\xx(t)=\mathbf{0}_N$ happens only a finite number of times with probability 1.
    \item The iterates of Algorithm \ref{commalgo} converge to $\xx_*$ with probability $1-\varepsilon(\{\eta\})$ where $\varepsilon(\{\eta\})\rightarrow 0$ as $\eta(0)\rightarrow 0$.
    \item The iterates of Algorithm \ref{commalgo} stay within an $\epsilon$-ball of $\xx_*$ after $T=\mathcal{O}(\log(1/\epsilon))$ with probability $1-\varepsilon(\{\eta\})$.
\end{enumerate}
\end{theorem}
Here the notation $\varepsilon(\{\eta\})$ denotes the dependence of the probability on the stepsize sequence $\{\eta(t)\}$. The dependence on general $\eta(0)$ has been omitted here for simplicity. But as an example, consider the stepsize sequence $\eta(t)=1/(t+t_0)^\mu$ for sufficiently large $t_0>0$ and $0.5<\mu\leq 1$, then $\varepsilon(\{\eta\})=\mathcal{O}\left(t_0^{1-\mu/2}\exp\left(-Ct_0^{\mu/2}\right)\right)$ for some constant $C>0$. 

\begin{algorithm}
\caption{\label{commalgo} \commalgoname Algorithm}
\textbf{Initialization: }Let $x_n(0)=0, \ \forall n\in\NN$ and $\eta(t)$ be the stepsize sequence. Let $\tlambda_n\sim\text{Unif}[\lambda_n,\lambda_n+\delta]$ for some $\delta>0$.

\textbf{At timesteps $t=0,1,\ldots$, each player $n\in\NN$}
\begin{enumerate}[label=(\arabic*)]
    \item Plays action $x_n(t)$ and observes a noisy reward $y_n(t)$.
    \item Updates their action as follows:
        \begin{equation}\label{iter-comm}
            x_n(t+1)=\Pi_{\XX_n}(x_n(t)+\eta(t)(\tlambda_n-y_n(t))),
        \end{equation}
        where $\Pi_{\XX_n}$ denotes the Euclidean projection into $[0,B_n]$.
    \item Transmits signal $s=1$ if $x_{n}(t+1)= B_n$, otherwise it does nothing (i.e., $s=0$). 
    \item Resets action to $0$, i.e., $x_n(t+1)=0$ upon receiving $s=1$.
\end{enumerate}
\textbf{End}
\end{algorithm}

\subsection{\nocommalgoname Algorithm}
If even the $1$-bit communication is not possible between the players, then a player cannot signal that it might be stuck at the boundary. If just that player resets its action to zero, then the resulting action vector might still be outside the domain of attraction of a desirable equilibrium. Hence, a reset mechanism is no longer an option with no communication. For scenarios like power control, where digital communication is not yet established between the devices, 1-bit signaling requires a special design that can become an implementation burden. 

Instead, in the fully distributed version of ToP, a player that is stuck at the boundary just projects their action back to the boundary, hoping that other players might help it reach the QoS later. This modification is detailed in Algorithm \ref{nocommalgo}. 

\begin{algorithm}
\caption{\label{nocommalgo} \nocommalgoname Algorithm}
\textbf{Initialization: }Let $x_n(0)=0, \ \forall n\in\NN$ and $\eta(t)$ be the stepsize sequence. Let $\tlambda_n\sim\text{Unif}[\lambda_n,\lambda_n+\delta]$ for some $\delta>0$.

\textbf{At timesteps $t=0,1,\ldots$, each player $n\in\NN$}
\begin{enumerate}[label=(\arabic*)]
    \item Plays action $x_n(t)$ and observes a noisy reward $y_n(t)$.
    \item Updates their action as follows:
        \begin{equation}\label{iter-nocomm}
            x_n(t+1)=\Pi_{\XX_n}(x_n(t)+\eta(t)(\tlambda_n-y_n(t))),
        \end{equation}
        where $\Pi_{\XX_n}$ denotes the Euclidean projection into $[0,B_n]$.

\end{enumerate}
\textbf{End}
\end{algorithm}

The following result gives guarantees for the \nocommalgoname algorithm.

\begin{theorem}\label{main-thm-nocomm}
     Under assumptions \ref{assu-martingale}, \ref{assu-delta}, and \ref{assu-step}, the iterates of Algorithm \ref{nocommalgo} converge with probability 1 to a point. The iterates of Algorithm \ref{nocommalgo} converge to $\xx_{*}$, as defined in Definition \ref{def:minimal}, with probability $1-\varepsilon(\{\eta\})$ where $\varepsilon(\{\eta\})\rightarrow 0$ as $\eta(0)\rightarrow 0$. Moreover, the iterates of Algorithm \ref{nocommalgo} stay within an $\epsilon$-ball of $\xx_*$ after $T=\mathcal{O}(\log(1/\epsilon))$ with probability $1-\varepsilon(\{\eta\})$.
\end{theorem}

Algorithm \ref{nocommalgo} also converges with probability 1. Unlike Algorithm \ref{commalgo}, the fully distributed variant may converge to a {\em bad equilibrium}, where one or more players are stuck at the boundary, i.e., $\exists n, s.t., \hat{x}_n=B_n$. However, the second part of Theorem \ref{main-thm-nocomm} states that with high probability (depending on the stepsize), not only does the algorithm avoid such a bad equilibrium, it actually converges to the best one possible. This ``best'' equilibrium satisfies the QoS condition, i.e., $u_n(\hat{\xx})=\tlambda_n$ for all players $n$ \textbf{and} is  ``minimal''.

\subsection{Meta Tug-of-Peace Algorithm}
We now present the Meta Tug-of-Peace (Meta-ToP) algorithm, a distributed algorithm that ensures convergence to an action profile where all players achieve their QoS requirements in Meta-ToW games. Similar to the ToP algorithm, this algorithm has the two key concepts of `cooperative' increase in actions and boundary signals. But here the boundary signals play two roles: avoiding getting stuck at the boundary, similar to ToP, \textit{and} switching game configurations. If a boundary signal is broadcast, it is a possible indicator that the QoS guarantees are not feasible under the current game configuration, and hence the players change the configuration. It is always possible that the configuration was feasible, and a player reached the boundary due to noise. But our algorithm ensures that each configuration will be tried infinitely often until the QoS requirements are met.

Meta-ToP works as follows: suppose player $n$ reaches the boundary $B_n$ at time $t$. It broadcasts $s=1$ to all players in the same game, i.e., all players $m$ with $g_m(t)=g_n(t)$ and signals $r=1$ to all players in all games. Each player who received signal $s=1$ switches a game with probability $\rho>0$. If a player $m$ decides to switch a game, they choose a game uniformly at random from $\GG\setminus\{g_m(t)\}$. Players which only receive signal $r=1$ perform the same procedure but with probability $\varphi>0$ instead of $\rho$. We call this procedure the \textit{game switch} mechanism. Each player $m$, irrespective of whether they switched their game or not, resets their action to $x_{m}(t+1)=0$ when they receive signal $r=1$.

\begin{algorithm}
\caption{\label{meta-algo} Meta Tug-of-Peace Algorithm}
\textbf{Initialization: }Let $x_n(0)=0,$ and $g_n(0)$ be chosen uniformly from $\GG$, $\forall n\in\NN$. Let
 $\eta(t)$ be the stepsize sequence, $\rho,\varphi$ be the switching probabilities, and $\tlambda_n\sim\text{Unif}[\lambda_n,\lambda_n+\delta]$ for some $\delta>0$.

\textbf{At timesteps $t=0,1,\ldots$, each player $n\in\NN$}
\begin{enumerate}[label=(\arabic*)]
    \item Plays action $x_n(t)$ in game $g_n(t)$ and observes a noisy reward $y_n(t)$.
    \item Updates their action as follows:
        \begin{equation*}
            x_n(t+1)=\Pi_{\XX_n}(x_n(t)+\eta(t)(\tlambda_n-y_n(t))),
        \end{equation*}
        where $\Pi_{\XX_n}$ denotes the Euclidean projection into $[0,B_n]$.
    \item If $x_n(t+1)=B_n$, then it transmits signal $s=1$ to all players in the same game, and $r=1$ to all players. Otherwise, it does nothing (i.e., $s=r=0$). 
    \item Resets action to $0$, i.e., $x_n(t+1)=0$ upon receiving $r=1$ from some player $m$.
    \item If signal $s=1$ is received, then $g_n(t+1)=g$ with probability $\rho/(K-1)$ for all $g\neq g_n(t)$ and $g_n(t+1)=g_n(t)$ with probability $1-\rho$.
    \item Else if only signal $r=1$ is received, then $g_n(t+1)=g$ with probability $\varphi/(K-1)$ for all $g\neq g_n(t)$ and $g_n(t+1)=g_n(t)$ with probability $1-\varphi$.
\end{enumerate}
\textbf{End}
\end{algorithm}

If a player in game $g$ reaches the boundary, then we would like more players to leave the game, hoping that the remaining players in the game will be able to achieve their QoS requirements. Hence, we want $\rho>\varphi$ for our algorithm. We require all players to have a positive probability of switching for technical reasons. Our analysis requires that if the game switch mechanism is performed infinitely many times, then all game configurations are chosen infinitely often with probability 1, irrespective of which player triggered the switch. An alternative simple mechanism that satisfies this requirement is trying all $K^N$ configurations sequentially, irrespective of which player triggered the switch. Similarly, another possible mechanism is where all players randomly decide their next game when they receive a boundary signal.

Our following result gives the convergence guarantees for the Meta-ToP algorithm:
\begin{theorem}\label{thm:MC}
    Under assumptions \ref{assu-martingale}, \ref{assu-delta}, and \ref{assu-step}, and for any $0<\rho,\varphi<1$, the action profiles given by Algorithm \ref{meta-algo}, $(\gbold(t),\xx(t))$, converge to an equilibrium point $(\hat{\gbold},\hat{\xx})$ which satisfies $u_n(\hat{\gbold},\hat{\xx})=\tlambda_n\geq \lambda_n$ for all $n$ with probability 1. Moreover, the game switches happen only a finite number of times with probability 1. 
\end{theorem}

Technically, any probabilities $0<\rho,\varphi<1$ guarantee that all game configurations are tried `sufficiently often'. Thus, the asymptotic convergence guaranteed by Theorem \ref{thm:MC} holds for any $0<\rho,\varphi<1$. However, $\rho,\varphi$ still affect the convergence rate. Conveniently, our simulations in Section \ref{sec:sim} show that convergence time is only mildly sensitive to these parameters.

\section{Convergence Analysis}\label{sec:analysis}
In this section, we present proofs for our convergence guarantees on the ToP algorithms. 
Our convergence analysis relies on stochastic approximation and the ODE method \cite{Borkar-book}. 

\subsection{Analysis of the ToP Algorithm for a Single ToW Game}
Consider the ODE:
\begin{equation}\label{ode}
\dot{\xx}(t)=h(\xx(t)),
\end{equation}
where $h_n(\xx(t))=\tlambda_n-u_n(\xx(t))$. Intuitively, ignoring resets and projections in iterations \eqref{iter-comm} and \eqref{iter-nocomm}, the ToP algorithms employ a noisy discretization of this ODE. We formally show how this ODE relates to our iterations later.  

By definition of a ToW game, this ODE satisfies  $\frac{\partial h_n}{\partial x_m}>0$ for all $n\neq m \in\NN$. Such an ODE\ is called a \textit{cooperative ODE} \cite{coop-1, coop-systems}. We have already assumed that there exists an equilibrium point for this ODE in the region $\XX^{o}$ (Assumption \ref{assu-delta}(a)). Then this class of ODEs has certain desirable convergence properties, which we restate in the following lemma:
\begin{lemma}\label{lemma:coop}
    The ODE in \eqref{ode} has the following properties:
    \begin{enumerate}[label=(\alph*)]
        \item \cite[Theorem 2.1]{coop-mini} For initial conditions in an open dense set, the solutions of \eqref{ode} converge to an equilibrium.
        \item \cite[Theorem 5.6]{coop-review} There exists a minimal equilibrium $\xx_*$ of \eqref{ode} such that any other equilibria $\hat{\xx}$ satisfies $x_{*_n}\leq \hat{x}_n$ for all $n\in\NN$.
        \item \cite{coop-systems} The dynamical system described by \eqref{ode} is monotone, i.e., if there are two solutions $\xx(\cdot)$ and $\xx'(\cdot)$ of \eqref{ode} with $\xx(0)\geq \xx'(0)$, then $\xx(t)\geq \xx'(t)$ for all $t\geq0$.
    \end{enumerate}
\end{lemma}
The first statement of the lemma implies that solutions of the ODE \eqref{ode}  converge to an equilibrium point that satisfies the QoS requirement. The last two statements together imply that any solution of the ODE initiated in the region $[\mathbf{0}_N,\xx_*]$ stays in $[\mathbf{0}_N,\xx_*]$ for all $t$. The following lemma shows that the equilibrium $\xx_*$ is a stable equilibrium with probability 1:
\begin{lemma}\label{lemma:stable}
    The equilibrium point $\xx_*$ is a stable equilibrium point for the ODE \eqref{ode} with probability 1. Moreover, it is a Locally Asymptotically Stable Equilibrium (LASE) \cite[B.3]{Borkar-book}.
\end{lemma}
\begin{proof}
Let $u(\cdot)$ be the vector function $(u_1(\cdot),\ldots,u_N(\cdot))^T$. Then $Dh(\xx)=-Du(\xx)$ for all $\xx$, where $Dh(\xx)$ and $Du(\xx)$ are the Jacobians of functions $h$ and $u$, respectively, evaluated at $\xx$. Sard's Lemma \cite{sard} states that the image of the set of points for which the Jacobian $Du(\cdot)$ is singular has Lebesgue measure zero, i.e., the set 
$$\ZZ=\{\mathbf z\in\RR^N:\exists \xx\in\RR^N\;\text{s.t.}\;u(\xx)=\mathbf z, \;\det(Du(\xx))=0\}$$
has Lebesgue measure zero. Since each coordinate of $\btlambda$ is chosen uniformly from the set $[\lambda_n,\lambda_n+\delta]$, $\btlambda\notin \ZZ$ a.s. This implies that $Dh(\xx_*)$ is nonsingular a.s.\ which implies that the equilibrium point $\xx_*$ is a.s.\ isolated. Using Assumption \ref{assu-delta}(b), the Jacobian $Dh(\xx_*)$ does not have purely imaginary eigenvalues a.s.\ and hence the equilibrium $\xx_*$ is hyperbolic a.s. Theorem 4.1.1 from \cite{Borkar-prospect} states that if the minimal equilibrium $\xx_*$ is hyperbolic, then it is a stable equilibrium. 

Now, since $\xx_*$ is isolated, there exists an open neighborhood $U\in\XX^o$ of $\xx_*$ such that $\xx_*$ is the only equilibrium in $U$. Then, because $\xx_*$ is stable, there exists a neighborhood $V\subseteq U$ such that every solution of \eqref{ode} initiated in $V$ remains in $U$ for all $t\geq 0$. Since every solution of \eqref{ode} converges to an equilibrium (Lemma \ref{lemma:coop}(a)), each solution starting in $V$ converges to an equilibrium in $U$. By uniqueness of the equilibrium in $U$, this limit must be $\xx_*$. Therefore, every solution starting in $V$ converges to $\xx_*$. Hence $\xx_*$ is a LASE.
\end{proof}
Lemma \ref{lemma:stable} implies that any solution of the ODE initiated in $[0,\xx_*]$ will converge to the equilibrium point $\xx_*$ and hence $\xx=\mathbf{0}_N$ is in the domain of attraction of the LASE $\xx_*$. We use this fact later in Lemma \ref{lemma:high-prob}. The next lemma, based on \cite[Section 5.1]{kushner-yin}, argues about the convergence of the iteration \eqref{iter-nocomm} in the FDToP algorithm to an ODE and therefore to equilibria.
\begin{lemma}\label{lemma:conv-boundary}
    Under assumptions \ref{assu-delta}-\ref{assu-step}, the following two statements hold:
    \begin{enumerate}[label=(\alph*)]
        \item The iterates of Algorithm \ref{nocommalgo} asymptotically track the solutions of the ODE
        \begin{equation}\label{boundary-ode}
        \dot{\xx}(t)=h(\xx(t))+b(\xx(t)),
        \end{equation}
        with probability 1. Here $b(\xx(t))=\textbf{0}$ in $[0,B_1)\times\ldots\times[0,B_N)$ and $b_n(\xx(t))=-h_n(\xx(t))$ if $x_n(t)=B_n$ and $h_n(\xx(t))>0$ for some $n$.
        \item With probability 1, the iterates of Algorithm \ref{nocommalgo} converge to some equilibrium point $\hat{\xx}$ which will fall under one of the following two cases: 
        \begin{itemize}
            \item $u_n(\hat{\xx})=\tlambda_n, \hat{x}_n<B_n$ for all players $n$,
            \item $\exists \;n \;\textrm{s.t.}\; \hat{x}_n=B_n$. 
        \end{itemize}
    \end{enumerate}
\end{lemma}
The term $b(\cdot)$ is the \textit{projection term} and is the force required to keep $\xx(t)$ inside $\XX$ at all times. The second part of the lemma states the existence of equilibria at the boundary, which may or may not satisfy our QoS condition. The projection term $b(\cdot)$ does not appear at $0$, the lower boundary, because $u_n(\xx)=0$ if $x_n=0$ and hence $h_n(\xx)>0$ at the boundary for $x_n=0$. This implies that the driving vector field of $h$ points inward at $0$, so $b(\xx)=0$ at the lower boundary. For Algorithm \ref{commalgo}, neither of the boundary projections has an impact. Just like the explanation above, the projection at the lower boundary has no effect, and the upper boundary projections are followed by a reset, which restart the iteration. 

The next lemma lower bounds the probability that the iterates of \eqref{iter-comm} or \eqref{iter-nocomm}  remain in a small ball around an equilibrium $\xx_*$ from some time onward, if $\xx(t)=\mathbf{0}_N$ for some $t$. 
\begin{lemma}\label{lemma:high-prob}
    For a system satisfying assumptions \ref{assu-delta}-\ref{assu-step}, large enough $t'$, and small enough $\epsilon$, there exists $T=\mathcal{O}(\log(1/\epsilon))$ such that the iterates of Algo.\ \ref{commalgo} and \ref{nocommalgo} satisfy the following:
    $$P(\|\xx(t)-\xx_*\|\leq \epsilon, \; \forall t\geq t'+T+1 \mid \xx(t')=\mathbf{0}_N)\geq 1-\phi(t'),$$
    where the probability sequence $\phi(t)$ satisfies $\sum_t \phi(t)<\infty$.
\end{lemma}
\begin{proof}
   Since $\xx_*$ lies in the interior of $\XX^o$, the boundary term in \eqref{boundary-ode} does not affect its stability properties. So, $\xx_*$ is a LASE for both ODEs \eqref{ode} and \eqref{boundary-ode} and $\mathbf{0}_N$ is in the domain of attraction for both these ODEs. As the map $h(\cdot)$ is continuously differentiable, and we have assumed that the Martingale difference sequence is bounded, we satisfy the assumptions of \cite{Borkar-conc}. The high probability result then follows from Theorem 1.1 of \cite{Borkar-conc}.
\end{proof}
The ``large enough'' values for $t',T$ and $1/\epsilon$ depend on the eigenvalues of the Jacobian at $\xx_*$, i.e., eigenvalues of the matrix $Dh(\xx_*)$. Since we care only about convergence, the only dependence that affects us is $t'$ as it relates to the stepsize choice. As a corollary, the result is valid from $t'=0$ onward if  $\eta(0)$ is small enough. Moreover, $\phi(0)\rightarrow 0$ as $\eta(0)\rightarrow 0$. 

 With the tools developed above, we can prove Theorem \ref{main-thm-comm}.
 
\begin{proof}[\textbf{Proof of Theorem \ref{main-thm-comm} (a)}]
    Let $\AAA(t')$ denote the negation of the set in the conditional probability in Lemma \ref{lemma:high-prob}: $$\AAA(t')=\{\exists \ t\geq t'+T+1 \ \text{s.t.} \|\xx(t)-\xx_*\|\geq \epsilon \}.$$ Then we know that $P(\AAA(t')\mid \xx(t')=\mathbf{0}_N)\leq \phi(t')$. The fact that $\phi(t)$ is summable implies that with probability 1 $$\sum_{t'} P(\AAA(t')\mid \xx(t')=\mathbf{0}_N)\II\{\xx(t')=\mathbf{0}_N\}<\infty.$$ Finally, through an extension of Borel-Cantelli Lemma \cite[Corollary 5.29]{Breiman}, we have that with probability 1:
    $$\sum_{t'}\II\{\AAA(t'),\xx(t')=\mathbf{0}_N\}<\infty.$$
    Hence, on the event $\{$\text{Resets happen infinitely often}$\}$, $\xx(t)$ from Algorithm \ref{commalgo} stays in an $\epsilon$-ball around $\xx_*$ after some time $t''$ with probability 1. Suppose the algorithm resets back to $\mathbf{0}_N$ infinitely often, then the iterates $\xx(t)$ would stay in an $\epsilon$-ball around $\xx_*$ after some time $t''$ with probability 1. This leads to a contradiction to our assumption that infinite resets happen. Therefore, the algorithm resets only finitely often. 
    
    Let $\tau$ denote the last such reset. Then for $t>\tau$, the iterates of \eqref{iter-comm} always stay in the interior of $\XX$. Then the iterates a.s.\ asymptotically track the solutions of the ODE \eqref{ode}. Hence, the actions of all players converge to an equilibrium point of the form $\hat{\xx}$ which satisfies $u_n(\hat{\xx})=\tlambda_n$ for all players $n$. 
\end{proof}

\begin{proof}[\textbf{Proof of Theorem \ref{main-thm-comm} (b) and (c)}]
       The algorithm is initialized at $\xx(0)=\mathbf{0}_N$, so for a small enough $\eta(0)$, Lemma \ref{lemma:high-prob} can be applied for $t'=0$. Thus, with probability $1-\phi(0)$, the iterates stay in the $\epsilon$-vicinity of $\xx_*$ from some $T$ onwards and converge to $\xx_*$ as $\xx_*$ is a LASE, where $T=\mathcal{O}(\log(1/\epsilon))$. We define $\varepsilon(\{\eta\})=\phi(0)$ and note that $\phi(0)\rightarrow 0$ as $\eta(0)\rightarrow 0$.
\end{proof}

The proof for Theorem \ref{main-thm-nocomm} is identical to the proof for part (b) of Theorem \ref{main-thm-comm}, but we repeat it for completeness. 

\begin{proof}[\textbf{Proof of Theorem \ref{main-thm-nocomm}}]
    Lemma \ref{lemma:conv-boundary} (b) implies that the iterates of Algorithm \ref{nocommalgo} converge to a point. As $\xx(0)=\mathbf{0}_N$, for small enough $\eta(0)$, Lemma \ref{lemma:high-prob} can be applied for $t'=0$. And hence with probability at least $1-\phi(0)$, the iterates stay in the $\epsilon$-vicinity of $\xx_*$ from some $T$ onwards and converge to $\xx_*$ as $\xx_*$ is a LASE, where $T=\mathcal{O}(\log(1/\epsilon))$. We define $\varepsilon(\{\eta\})=\phi(0)$ and note that $\phi(0)\rightarrow 0$ as $\eta(0)\rightarrow 0$. 
\end{proof}

\subsection{Analysis of the Meta-ToP Algorithm}
\begin{proof}[\textbf{Proof of Theorem \ref{thm:MC}}]
Define $\hat{\GG}=\{\gbold\in\GG^N\mid \exists \xx\in\intXX \; \text{s.t.}\;  u_n(\gbold,\xx)=\tlambda_n, \forall n\in\NN\}$, i.e., the set of game configurations in which an action profile exists such that QoS requirements of all players are satisfied. Then using Lemma \ref{lemma:coop}, all configurations $\hat{\gbold}\in\hat{\GG}$ are associated with a minimal equilibrium $\xx_*(\hat{\gbold})$. 

Then Lemma \ref{lemma:high-prob} shows that for $\hat{\gbold}\in\hat{\GG}$, and large enough $t',T$, and small enough $\epsilon$, 
    \begin{align*}
        &P\Big(\|\xx_g(t)-\xx_{*_g}(\hat{\gbold})\|\leq \epsilon, \; \forall t\geq t'+T+1\\
        &\;\;\;\;\;\;\;\;\;\;\;\;\;\;\;\;\;\;\;\;\;\;\Big| \; \xx_g(t')=\mathbf{0}_N,\gbold(t')=\hat{\gbold}\Big)\geq 1-\phi(t'),
    \end{align*}
    holds for all games $g\in\GG$. The probability sequence $\phi(t)$ satisfies $\sum_t \phi(t)<\infty$. Here $\xx_g(t)$ denotes the vector of actions for all players who are in game $g$ at time $t$, and $\xx_{*_g}(\hat{\gbold})$ denotes the subvector of $\xx_*(\hat{\gbold})$ for players in game $g$. Applying the union bound over all games gives us
        \begin{align*}
        &P\Big(\|\xx(t)-\xx_*(\hat{\gbold})\|\leq \epsilon\sqrt{K}, \; \forall t\geq t'+T+1\\
        &\;\;\;\;\;\;\;\;\;\;\;\;\;\;\;\;\;\;\;\;\;\;\Big| \; \xx(t')=\mathbf{0}_N,\gbold(t')=\hat{\gbold}\Big)\geq 1-K\phi(t').
    \end{align*}
    Now, similar to the argument in Theorem \ref{main-thm-comm}, let $\AAA_{\hat{\gbold}}$ be the set.
    $$\AAA_{\hat{\gbold}}(t')=\{\exists \ t\geq t'+T+1 \ \text{s.t.} \|\xx(t)-\xx_*(\hat{\gbold})\|\geq \epsilon\sqrt{K} \}.$$ Then $P(\AAA_{\hat{\gbold}}(t')\mid \xx(t')=\mathbf{0}_N, \gbold(t')=\hat{\gbold})\leq K\phi(t')$. The fact that $\phi(t)$ is summable implies that with probability 1 
    \begin{align*}
        &\sum_{t'} P\left(\AAA_{\hat{\gbold}}(t')\mid \xx(t')=\mathbf{0}_N,\gbold(t')=\hat{\gbold}\right)\\
        &\;\;\;\;\;\;\;\;\;\;\;\;\;\;\;\;\;\;\;\;\;\;\;\;\;\;\;\;\;\times \II\{\xx(t')=\mathbf{0}_N,\gbold(t')=\hat{\gbold}\}<\infty.
    \end{align*} Finally, through an extension of Borel-Cantelli Lemma \cite[Corollary 5.29]{Breiman}, we have that with probability 1:
    $$\sum_{t'}\II\{\AAA_{\hat{\gbold}}(t'),\xx(t')=\mathbf{0}_N,\gbold(t')=\hat{\gbold}\}<\infty.$$
    Hence, on the event \{Game configuration switches to $\hat{\gbold}$ infinitely often\}, $\xx(t)$ from Algorithm \ref{meta-algo} stays in an $\epsilon$-ball around $\xx_*(\hat{\gbold})$ after some time $t''$ with probability 1.

    Irrespective of which player triggers the switch at time $t$, the probability of reaching configuration $\gbold$ at time $t+1$ is bounded away from zero for all games $\gbold\in \GG^N$. And hence, if game switches happen infinitely often, all game configurations are chosen infinitely often. This implies that the configuration $\hat{\gbold}$ is chosen infinitely often if game switches happen infinitely often. Thus, the iterates $\xx(t)$ would stay in an $\epsilon$-ball around $\xx_*(\hat{\gbold})$ after some time $t''$ with probability 1, which contradicts our assumption that switches happen infinitely often. 
    
    Let $\tau$ denote the last such switch, and $\gbold(\tau)$ be the final game configuration. Note that $\gbold(\tau)\in\hat{\GG}$. Then for $t>\tau$, the actions $\xx(t)$ always stay in the interior of $\XX$. Then for each individual game $g$, the actions of players in that game a.s.\ asymptotically track the solutions of the ODE for that specific game $\dot{\xx}_{g}(t)=h^{(g)}(\xx_g(t))$ which satisfies Lemma \ref{lemma:coop}. Hence, all players converge to an equilibrium point of the form $\hat{\xx}$ which satisfies $u_n(\gbold(\tau),\hat{\xx})=\tlambda_n$ for all players $n$. This completes the proof for Theorem \ref{thm:MC}.
\end{proof}

\section{Simulations}\label{sec:sim}

We simulated the games from  Section \ref{sec:app}. Except for single-run plots, utility results are averaged over 100 independent runs, with the shaded region indicating the interquartile range (region between the first and third quartiles). Plots showing the percentage of converged runs depict the fraction of runs that have converged by time~$t$, computed over 500 runs.

\begin{figure*}[t]
\centering
\begin{subfigure}{.3\textwidth}
  \centering
  \includegraphics[width=0.99\linewidth]{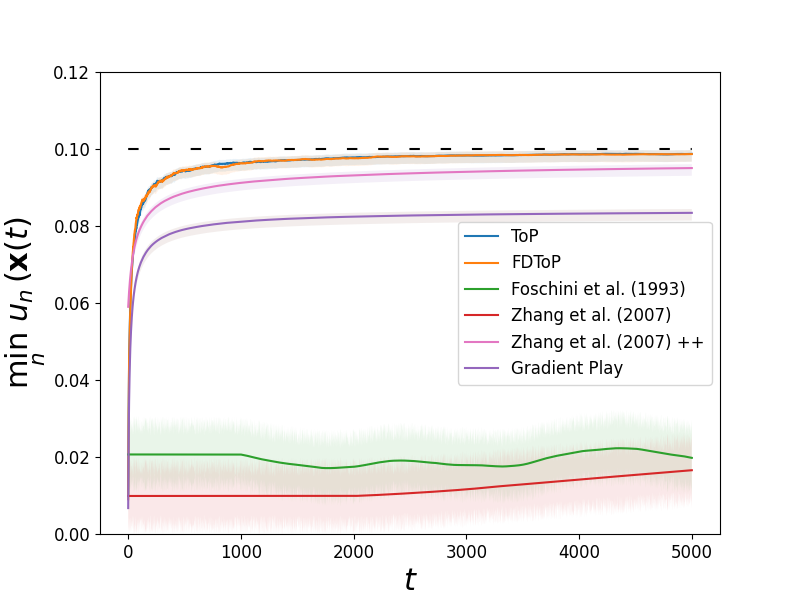}
  \caption{Comparison between algorithms}
  \label{fig:pc_comp}
\end{subfigure}%
\begin{subfigure}{.3\textwidth}
  \centering
  \includegraphics[width=.99\linewidth]{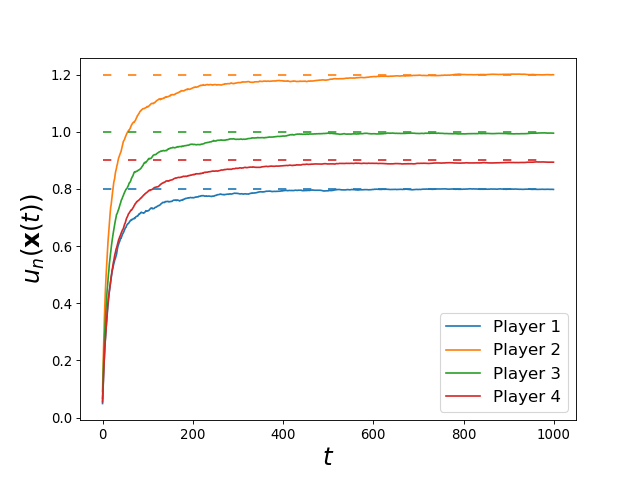}
  \caption{Performance of Tug-of-Peace Algorithm}
  \label{fig:pc_4}
\end{subfigure}
\begin{subfigure}{.3\textwidth}
  \centering
  \includegraphics[width=.99\linewidth]{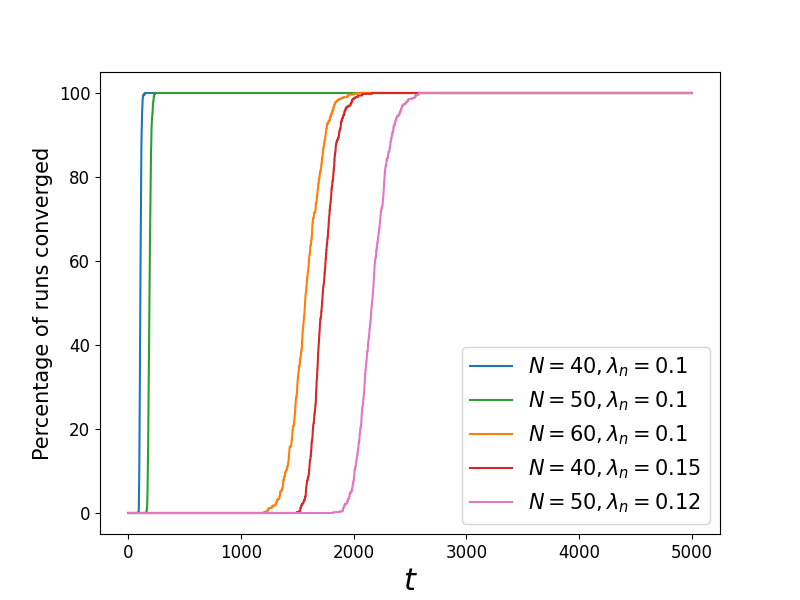}
  \caption{Percentage of runs converged}
  \label{fig:pc_ECDF}
\end{subfigure}
\caption{Single-channel power control game with (a) $N=50$, (b) $N=4$ players, and (c) varying number of players}
\label{fig:pc}
\end{figure*} 
\begin{figure*}[t]
\centering
\begin{subfigure}{.3\textwidth}
  \centering
  \includegraphics[width=0.99\linewidth]{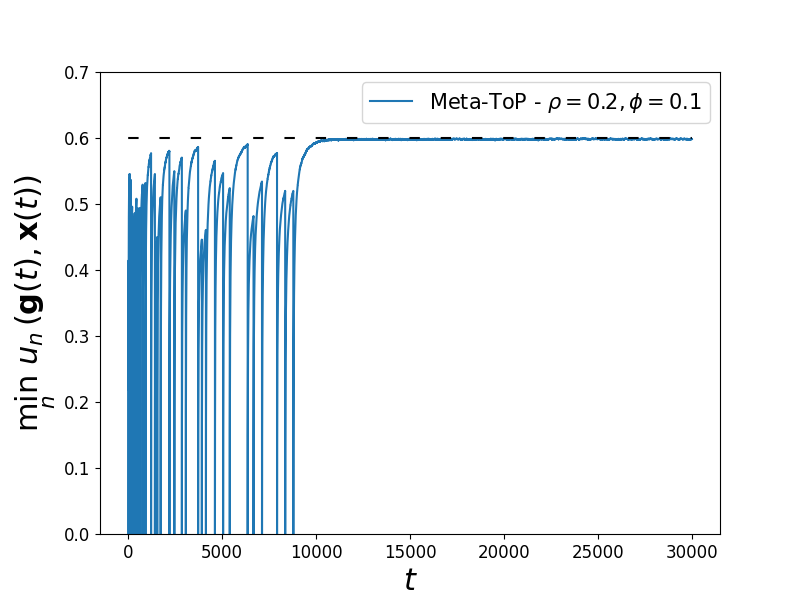}
  \caption{Convergence of Meta-ToP to QoS}
  \label{fig:multi_channel_single}
\end{subfigure}%
\begin{subfigure}{.3\textwidth}
  \centering
  \includegraphics[width=.99\linewidth]{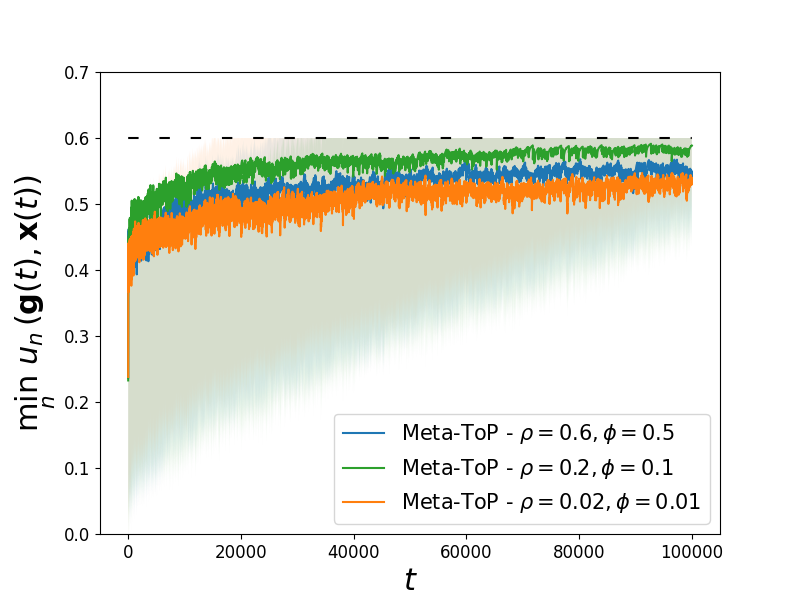}
  \caption{Effect of switching probabilities}
  \label{fig:multi_channel_avg}
\end{subfigure}
\begin{subfigure}{.3\textwidth}
  \centering
  \includegraphics[width=.99\linewidth]{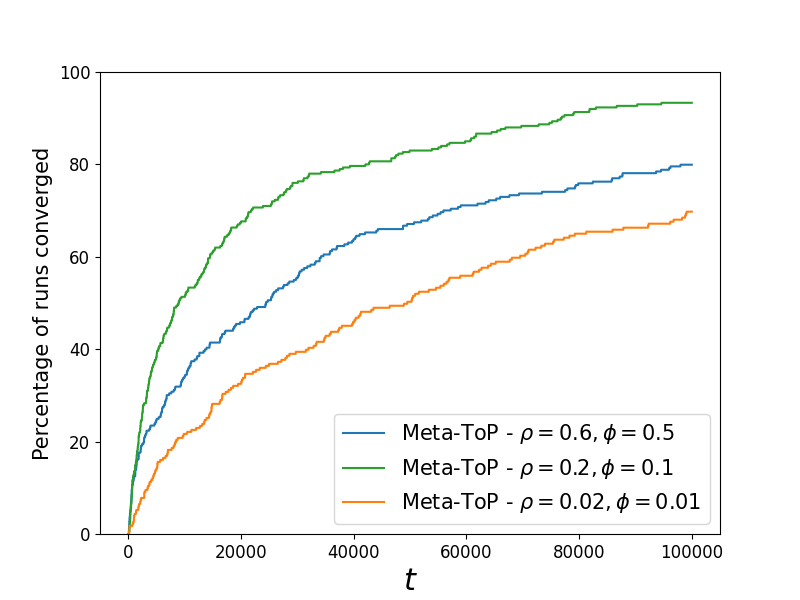}
  \caption{Percentage of runs converged}
  \label{fig:multi_channel_ECDF}
\end{subfigure}
\caption{Performance of Meta-ToP algorithm in multi-channel power control}
\label{fig:multi_channel}
\end{figure*} 

\subsection{Simulations for Power Control}
For the power control game, we first simulate our ToP algorithm for the single-channel case. We randomly generate a diagonal heavy channel gain matrix $C$ where the diagonal elements are uniformly and independently sampled from $[0.2,0.8]$ and the non-diagonal elements are uniformly and independently sampled from $[0,0.2]$. We set $N_0=0.1$. Fig. \ref{fig:pc_comp} compares our algorithms with power control algorithms with $N=50$ players and $\lambda_n=0.1$ for each player. We plot $\min_{n\in\NN} u_n(\xx(t))$, i.e., the reward of the player with the minimum reward or the reward of the most dissatisfied player at each time. The rewards observed are noisy with additive Gaussian measurement noise $\mathcal{N}(0,0.1)$. The upper boundary $B_n$ is set to be large ($B_n=1$) so that it has no effect. We use the stepsize sequence $\eta(t)=1/(t+100)$ for our algorithm. The dashed lines represent the QoS requirements.

The algorithms by \cite{power_foschini} (\textit{Foschini et al.} in Fig. \ref{fig:pc_comp}) and \cite{power_zhang} (\textit{Zhang et al.} in Fig. \ref{fig:pc_comp}) fail to converge in our case. While \cite{power_foschini} cannot handle noise, \cite{power_zhang} assumes an unbiased estimator of the inverse SINR, i.e., the inverse reward, which is unavailable to our algorithms. The curve \textit{Zhang et al.\ (2007)$++$} is of the setting in \cite{power_zhang}, i.e., where Gaussian noise is added to the inverse SINR. The curve \textit{Gradient Play} corresponds to the algorithm where each player performs gradient ascent on their own utility. Gradient play converges to an equilibrium that does not satisfy the QoS guarantees in this example. As expected, our algorithms converge to the minimal point that satisfies the QoS requirements.

Fig. \ref{fig:pc_4} plots the utilities of each player using the ToP algorithm for a power control game with $N=4$ and different QoS requirements $\blambda=[0.8,1.2,1,0.9]^T$. We use the stepsize sequence $\eta(t)=1/(t+10)^{0.9}$. The different dashed lines are the QoS requirements. The actions eventually converge to a profile that satisfies the QoS requirements for each player.

In Fig.\ \ref{fig:pc_ECDF}, we plot the rate of convergence for the ToP algorithm for different numbers of players and different QoS requirements. As expected, it takes longer for the players to converge as the number of players and the QoS requirements increase. The first two plots ($N=40, \lambda_n=0.1$ and $N=50, \lambda_n=0.1$) show that these problems are easier to solve and the players almost immediately converge to an action profile that satisfies the QoS requirements. The remaining QoS requirements are near the boundary of the feasibility region; no action profiles can achieve a utility of $0.16$ for all $N=40$ players, $0.13$ for all $N=50$ players, or $0.11$ for all $N=60$ players; hence, convergence takes longer.

For the multi-channel case (Fig.\ \ref{fig:multi_channel}), we simulate the Meta-ToP algorithm for $N=100$ players, $K=10$ channels, and $\lambda_n=0.6$ for each player. The rest of the parameters are set as in  Fig.\ \ref{fig:pc_comp}. In Fig.\ \ref{fig:multi_channel_single}, we see convergence of the Meta-ToP algorithm for a single run. There are multiple game switches in the beginning, leading to frequent returns to zero utility. A feasible game configuration is eventually obtained and the players' utilities converge to the QoS requirement.

 Fig.\ \ref{fig:multi_channel_avg} shows the performance of the Meta-ToP algorithm for three different switching probabilities $\rho, \varphi$. The best performance is obtained for  $\rho=0.2, \varphi=0.1$. Significantly larger probabilities ($\rho=0.6, \varphi=0.5$) lead to worse performance since almost random game allocations are being tried after each switch. Smaller probabilities ($\rho=0.02, \varphi=0.01$) also have worse performance as this results in very slow exploration. We study this in more depth in Fig.\ \ref{fig:multi_channel_ECDF}, where we plot the rate of convergence for different switching probabilities.

 \begin{figure*}[t]
\centering
\begin{subfigure}{.3\textwidth}
  \centering
  \includegraphics[width=0.99\linewidth]{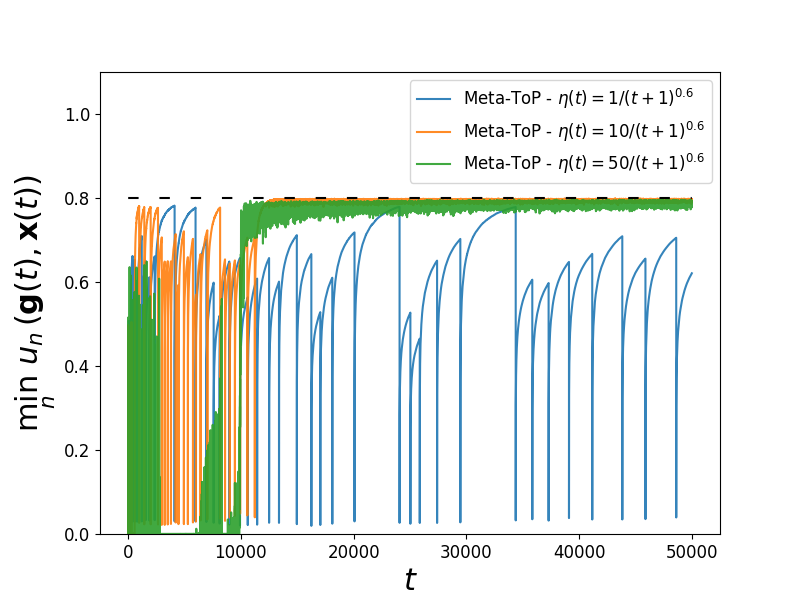}
  \caption{Effect of stepsize}
  \label{fig:task_stepsize}
\end{subfigure}%
\begin{subfigure}{.3\textwidth}
  \centering
  \includegraphics[width=.99\linewidth]{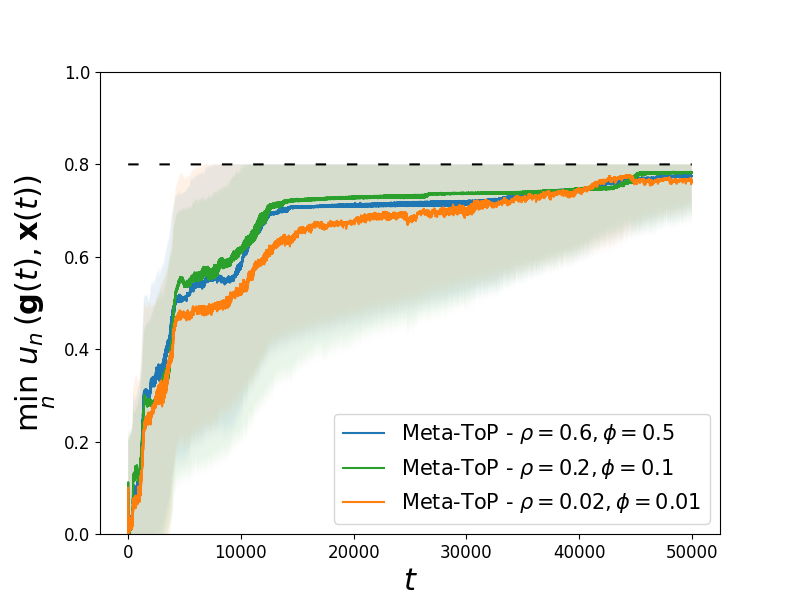}
  \caption{Effect of switching probabilities}
  \label{fig:task_probs}
\end{subfigure}
\begin{subfigure}{.3\textwidth}
  \centering
  \includegraphics[width=.99\linewidth]{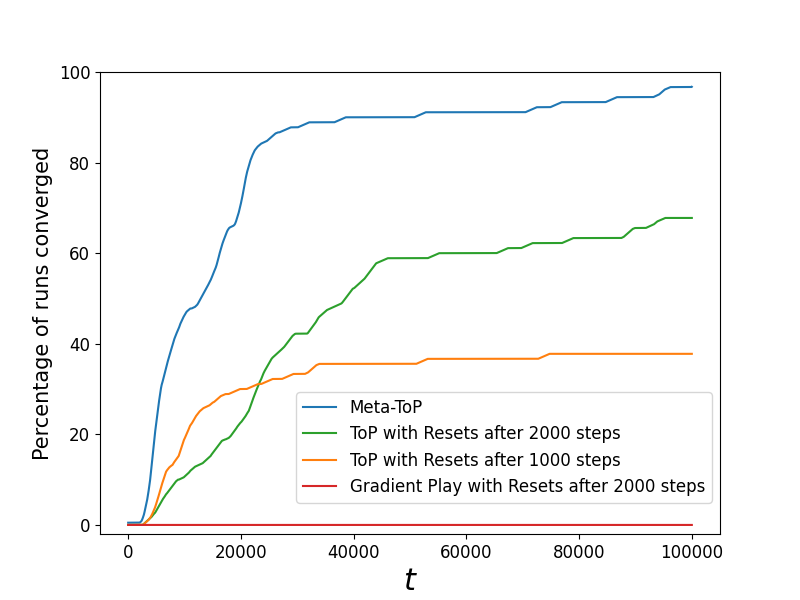}
  \caption{Percentage of runs converged}
  \label{fig:task_ECDF}
\end{subfigure}
\caption{Distributed task allocation for (a) Meta-ToP with $\rho=0.2,\varphi=0.1$, (b) Meta-ToP with $\eta(t)=10/(t+1)^{0.6}$, and (c) various algorithms}
\label{fig:task}
\end{figure*} 
\begin{figure*}[t]
\centering
\begin{subfigure}{.3\textwidth}
  \centering
  \includegraphics[width=0.99\linewidth]{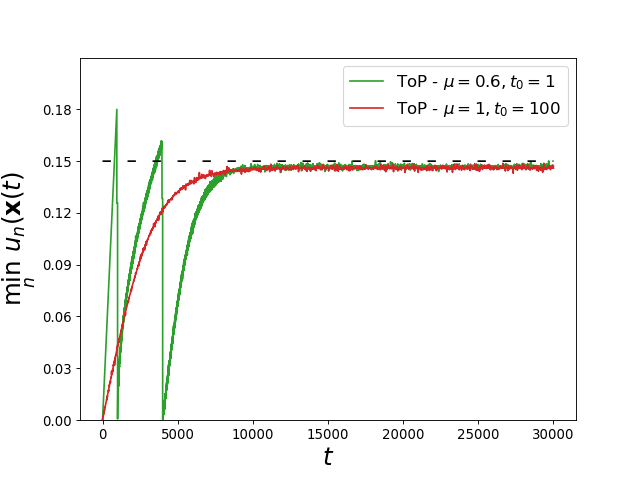}
  \caption{Effect of stepsize on ToP algorithm}
  \label{fig:sa_comm}
\end{subfigure}%
\begin{subfigure}{.3\textwidth}
  \centering
  \includegraphics[width=.99\linewidth]{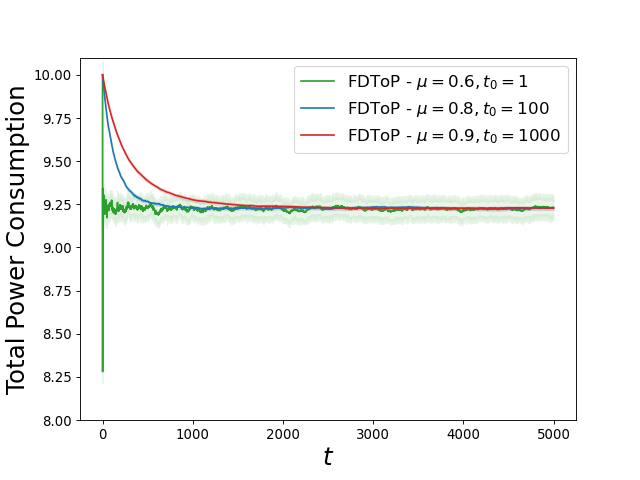}
  \caption{Effect of stepsize on FDToP algorithm}
  \label{fig:sa_nocomm}
\end{subfigure}
\begin{subfigure}{.3\textwidth}
  \centering
  \includegraphics[width=.99\linewidth]{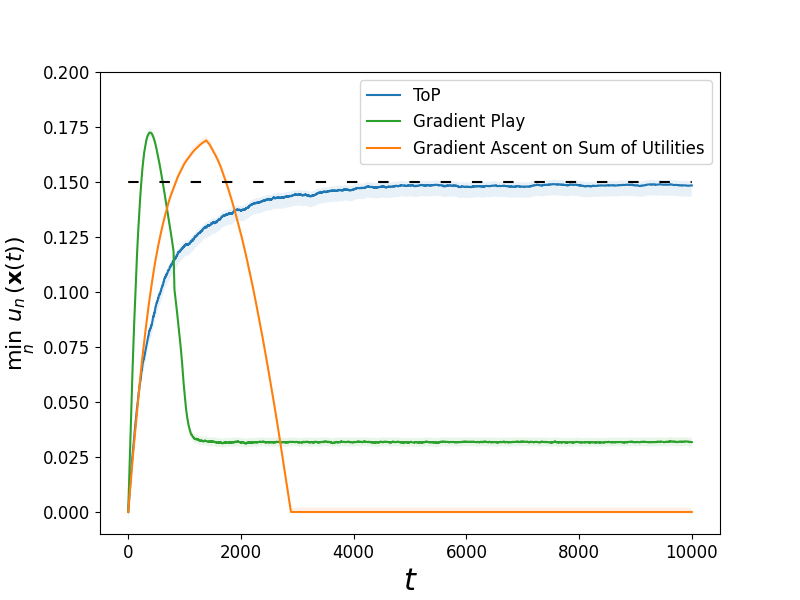}
  \caption{Comparison between algorithms}
  \label{fig:sa_compare}
\end{subfigure}
\caption{Performance of various algorithms in sensor activation game with stepsize $\eta(t)=1/(t+t_0)^{\mu}$}
\label{fig:sa}
\end{figure*} 

\subsection{Simulations for Distributed Task Allocation}

For distributed task allocation, we consider $N=100$ agents and $K=10$ tasks. The agents' efforts are assumed to be in $[0,10]$. The constants $\alpha_g$ and $\beta_{m,g}$ for the utility function are uniformly and independently sampled from $[1.1,5]$ and $[100,200]$, respectively. The QoS requirement for each player is $0.8$, i.e., $\lambda_n=0.8$ for all $n\in\NN$. We study the effect of stepsize sequences and the effect of switching probabilities on the performance of the Meta-ToP algorithm (Figure \ref{fig:task}).

Fig.\ \ref{fig:task_stepsize} shows the performance of the Meta-ToP algorithm in a single run for three different stepsize sequences. For a smaller stepsize $(\eta(t)=1/(t+1)^{0.6})$, it takes the players longer to reach boundaries, so they stay in infeasible game configurations for longer. Hence, fewer game switches happen, and convergence takes longer. For larger stepsizes $(\eta(t)=10/(t+1)^{0.6})$, game switches happen quicker and a game configuration in which QoS requirements can be met is found earlier. This leads to faster convergence. The convergence time did not improve significantly for larger stepsizes $(\eta(t)=50/(t+1)^{0.6})$. Although players try more game configurations, they reach the boundary before the algorithm can converge to an action profile that satisfies the QoS requirements, even in game configurations where it is feasible.

Fig.\ \ref{fig:task_probs}, shows performance of the Meta-ToP algorithm for three different switching probabilities. Unlike in the power control (Fig.\ \ref{fig:multi_channel_avg}), performance does not vary significantly with $\rho, \varphi$. The performance for $\rho=0.2, \varphi=0.1$ is still slightly better, but the convergence rate is similar in all three cases.

In Fig.\ \ref{fig:task_ECDF}, we compare the rate of convergence of our algorithm with two benchmarks. For our Meta-ToP algorithm, close to 90\% of runs converge by $t=25,000$, and 99\% of runs converge by $t=100,000$. In one benchmark algorithm, the actions are updated using the ToP algorithm, but the switching happens regularly in fixed intervals. Every $\tau$ timesteps, if not all players have achieved their QoS requirements, a new game configuration is randomly chosen. We tried this algorithm with multiple values of $\tau$ and selected the best-performing ones. While some runs converge for this algorithm, the performance is significantly worse than our Meta-ToP algorithm. As the second benchmark, we compare with an algorithm where actions are updated via gradient play, and switches are made in fixed intervals. This algorithm does not converge to an action profile where all players achieve their QoS requirements.

\subsection{Simulations for Activation in Sensor Networks}

 For the sensor activation game ($K=1$), we simulate a network with $N=10$ sensors. Each sensor has a set of multiple routes for sending packets to the destination, and we assume that a transmission is successful if there exists a route where all sensors are active. For each simulation, we randomly generate an Erdős–Rényi graph with edge probability $0.2$ for the sensors. This graph dictates all possible paths from each sensor to its destination. We choose $L=100$ packets and $f(p)=0.8\sqrt{p}$, $\alpha=0.8$ and $\beta=0.7$. The QoS requirements are chosen as $\lambda_n=0.15$ for all players $n\in\NN$.

In the simulations for the single-channel power control game (Subsection VI.A), the actions never got stuck at the boundary, and the equilibrium is unique. Here, since multiple equilibria exist, we studied the effects of the stepsize sequence on the resulting equilibrium (Fig. \ref{fig:sa}).

In Fig. \ref{fig:sa_comm}, we plot the performance of the ToP algorithm. The first plot is for a stepsize sequence which decreases slowly and has a higher initial value, i.e., $\eta(t)=1/(t+1)^{0.6}$. In this setting, the rewards quickly increase in the beginning but one of the players reaches the boundary, causing a reset. This happens twice before the algorithm stabilizes and converges to the minimal equilibrium. The second plot is for a quickly decreasing stepsize sequence with a lower initial value $(\eta(t)=1/(t+100))$. In this case, the actions slowly increase and directly reach the minimal equilibrium point. 

The second plot (Fig. \ref{fig:sa_nocomm}) shows the total power consumed by all sensors when they run the FDToP algorithm with different stepsizes. Stepsizes with higher initial values start with large oscillations but eventually, reach close to the minimal equilibrium point faster. On the other hand, quickly decreasing stepsizes with lower initial values are more stable but take longer to converge. Nevertheless, for all different stepsize sequences, the algorithm still converged to the same equilibrium, which was the minimal equilibrium. 

In the sensor activation game, there are many more equilibrium points of our ODEs, including many at the boundary. But in the extensive simulations we performed, the algorithm always converged to the minimal equilibrium point. We tried various levels of noise (with independent additive noise), and stepsize sequences that go down very slowly (e.g., $\eta(t)=1/\lceil{t/100\rceil}^{0.51}$). Even in these conditions, both our algorithms converged to the minimal equilibrium point, despite the initial oscillations. These empirical findings suggest that in special cases of ToW games, the convergence properties can be stronger than our general theoretical guarantees. 

In Fig.\ \ref{fig:sa_compare}, we compare the performance of the ToP algorithm with two gradient-based algorithms. The first algorithm is gradient play, where each player updates their action using gradient ascent on their own utility. The second algorithm is where players update their action using the gradient of the sum of utilities. While our algorithm converges to an action profile where each player achieves their QoS requirements, the gradient-based algorithms converge to a point where at least one player achieves significantly less than their QoS requirement. In both gradient-based algorithms, most players either chose the action $B=1$ or the action $0$ at equilibrium.

%While we do not have a theoretical statement, this observation is interesting empirically.

\section{CONCLUSIONS}
We studied a class of meta games where players decide the game they wish to participate in, and their action in that game. The rewards of other players in the same game decrease when a player increases their effort. We therefore called these games `Tug-of-War'' (ToW) games, which can model wireless power control, distributed task allocation, and sensor activation.

We proposed a simple stochastic approximation algorithm that players can use to converge to a point that satisfies their QoS requirements if such a point is feasible. We proved that for Meta-ToW games, our simple algorithm (Meta-ToP) converges to such a desirable point with probability 1. Moreover, for single ToW games, the algorithm converges to the point where the players' actions are ``minimal'' with high probability, which is useful when the action represents power or energy. We demonstrate the performance of our algorithm through simulations.

Analyzing ToW games with asynchronous players can remove the dependence on slotted time. It is also an open question what the fastest convergence in ToW games is and whether the ToP achieves it. Studying better mechanisms of switching between games is another interesting direction.

\section*{References}
\bibliographystyle{ieeetr}
\bibliography{root.bib}

\end{document}